\documentclass{article}%
\usepackage{amsmath}
\usepackage{amsfonts}
\usepackage{amssymb}
\usepackage{graphicx}
\usepackage{setspace}%
\providecommand{\U}[1]{\protect\rule{.1in}{.1in}}
\newtheorem{theorem}{Theorem} [section]

\newtheorem{corollary}[theorem]{Corollary}

\newtheorem{definition}[theorem]{Definition}

\newtheorem{lemma}[theorem]{Lemma}

\newtheorem{problem}[theorem]{Problem}

\newenvironment{proof}[1][Proof]{\noindent\textbf{#1.} }{\ \rule{0.5em}{0.5em}}
\begin{document}

\author{Vadim E. Levit\\Ariel University Center of Samaria, ISRAEL\\E-mail: levitv@ariel.ac.il
\and Eugen Mandrescu\\Holon Institute of Technology, ISRAEL\\E-mail: eugen\_m@hit.ac.il }
\date{}
\title{Very Well-Covered Graphs of Girth at least Four and Local Maximum Stable Set
Greedoids\thanks{A preliminary version of this paper has been presented at the
38$^{th}$ Southeastern International Conference on Combinatorics, Graph
Theory, and Computing, March 5-9, 2007, Boca-Raton, Florida, USA.}}
\maketitle

\begin{abstract}
A \textit{maximum stable set }in a graph $G$ is a stable set of maximum
cardinality. $S$ is a \textit{local maximum stable set} of $G$, and we write
$S\in\Psi(G)$, if $S$ is a maximum stable set of the subgraph induced by
$S\cup N(S)$, where $N(S)$ is the neighborhood of $S$.

Nemhauser and Trotter Jr. \cite{NemhTro}, proved that any $S\in\Psi(G)$ is a
subset of a maximum stable set of $G$. In \cite{LevMan2} we have shown that
the family $\Psi(T)$ of a forest $T$ forms a greedoid on its vertex set. The
cases where $G$ is bipartite, triangle-free, well-covered, while $\Psi(G)$ is
a greedoid, were analyzed in \cite{LevMan45}, \cite{LevMan07}, \cite{LevMan08}%
, respectively.

In this paper we demonstrate that if $G$ is a very well-covered graph of girth
$\geq4$, then the family $\Psi(G)$ is a greedoid if and only if $G$ has a
unique perfect matching.


\end{abstract}

\section{Introduction}

Throughout this paper $G=(V,E)$ is a simple (i.e., a finite, undirected,
loopless and without multiple edges) graph with vertex set $V=V(G)$ and edge
set $E=E(G)$. If $X\subset V$, then $G[X]$ is the subgraph of $G$ spanned by
$X$. $K_{n},C_{n}$ denote respectively, the complete graph on $n\geq1$
vertices and the chordless cycle on $n\geq3$ vertices. If $A,B$ $\subset V$
and $A\cap B=\emptyset$, then $(A,B)$ stands for the set $\{e=ab:a\in A,b\in
B,e\in E\}$.

The \textit{neighborhood} of a vertex $v\in V$ is the set $N(v)=\{u:u\in V$
\ \textit{and} $vu\in E\}$. For $A\subset V$, we denote $N(A)=\{v\in
V-A:N(v)\cap A\neq\emptyset\}$ and $N[A]=A\cup N(A)$.

A \textit{stable} set in $G$ is a set of pairwise non-adjacent vertices. A
stable set of maximum size will be referred to as a \textit{maximum stable
set} of $G$, and the \textit{stability number }of $G$, denoted by $\alpha(G)$,
is the cardinality of a maximum stable set in $G$. Let $\Omega(G)$ stand for
the set of all maximum stable sets of $G$.

A \textit{matching} in a graph $G=(V,E)$ is a set of edges $M\subseteq E$ such
that no two edges of $M$ share a common vertex. A \textit{maximum matching} is
a matching of maximum cardinality. By $\mu(G)$ is denoted the cardinality of a
maximum matching. A matching is \textit{perfect} if it saturates all the
vertices of the graph.

Let us recall that $G$ is a \textit{K\"{o}nig-Egerv\'{a}ry graph} provided
$\alpha(G)+\mu(G)=\left\vert V(G)\right\vert $ \cite{Dem}, \cite{Ster}. As a
well-known example, any bipartite graph is a K\"{o}nig-Egerv\'{a}ry graph
\cite{Eger}, \cite{Koen}.

\begin{theorem}
\cite{LevMan2003}\label{th4} If $G$ is a \textit{K\"{o}nig-Egerv\'{a}ry graph,
then every maximum matching is contained in }$\left(  S,V(G)-S\right)  $, for
each $S\in\Omega\left(  G\right)  $.
\end{theorem}

A matching $M=\{a_{i}b_{i}:a_{i},b_{i}\in V(G),1\leq i\leq k\}$ of graph $G$
is called \textit{a uniquely restricted matching} if $M$ is the unique perfect
matching of $G[\{a_{i},b_{i}:1\leq i\leq k\}]$ \cite{GolHiLew}. For instance,
all the maximum matchings of the graph $G$ in Figure \ref{fig111} are uniquely
restricted, while the graph $H$ from the same figure has both uniquely
restricted maximum matchings (e.g., $\{uv,xw\}$) and non-uniquely restricted
maximum matchings (e.g., $\{xy,tv\}$).

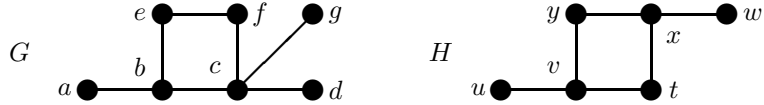
\begin{figure}[h]
\setlength{\unitlength}{1.0cm} \begin{picture}(5,1.5)\thicklines
\multiput(3,0)(1,0){4}{\circle*{0.29}}
\multiput(4,1)(1,0){3}{\circle*{0.29}}
\put(3,0){\line(1,0){3}}
\put(5,0){\line(1,1){1}}
\put(4,1){\line(1,0){1}}
\multiput(4,0)(1,0){2}{\line(0,1){1}}
\put(2.7,0){\makebox(0,0){$a$}}
\put(3.7,0.3){\makebox(0,0){$b$}}
\put(4.7,0.3){\makebox(0,0){$c$}}
\put(6.3,0){\makebox(0,0){$d$}}
\put(3.7,1){\makebox(0,0){$e$}}
\put(5.3,1){\makebox(0,0){$f$}}
\put(6.3,1){\makebox(0,0){$g$}}
\put(2.1,0.5){\makebox(0,0){$G$}}
\multiput(8.5,0)(1,0){3}{\circle*{0.29}}
\multiput(9.5,1)(1,0){3}{\circle*{0.29}}
\put(8.5,0){\line(1,0){2}}
\put(9.5,1){\line(1,0){2}}
\multiput(9.5,0)(1,0){2}{\line(0,1){1}}
\put(8.2,0){\makebox(0,0){$u$}}
\put(9.2,0.3){\makebox(0,0){$v$}}
\put(10.8,0){\makebox(0,0){$t$}}
\put(11.85,1){\makebox(0,0){$w$}}
\put(9.2,1){\makebox(0,0){$y$}}
\put(10.8,0.7){\makebox(0,0){$x$}}
\put(7.7,0.5){\makebox(0,0){$H$}}
\end{picture}\caption{The unique cycle of $H$ is alternating with respect to
the matching $\{yv,tx\}$.}%
\label{fig111}%
\end{figure}

Recall that $G$ is \textit{well-covered} if all its maximal stable sets have
the same cardinality \cite{Plummer}, and $G$ is \textit{very well-covered} if,
in addition, it has no isolated vertices and $\left\vert V(G)\right\vert
=2\alpha(G)$ \cite{Favaron}. 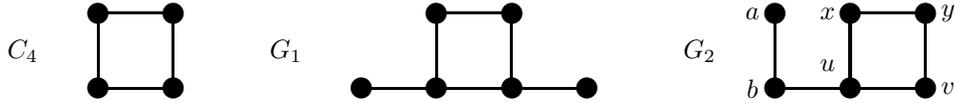
\begin{figure}[h]
\setlength{\unitlength}{1cm}\begin{picture}(5,1.4)\thicklines
\multiput(2,0)(1,0){2}{\circle*{0.29}}
\multiput(2,1)(1,0){2}{\circle*{0.29}}
\put(2,0){\line(1,0){1}}
\put(2,1){\line(1,0){1}}
\put(2,0){\line(0,1){1}}
\put(3,0){\line(0,1){1}}
\put(1,0.5){\makebox(0,0){$C_4$}}
\multiput(5.5,0)(1,0){4}{\circle*{0.29}}
\multiput(6.5,1)(1,0){2}{\circle*{0.29}}
\multiput(6.5,0)(1,0){2}{\line(0,1){1}}
\put(5.5,0){\line(1,0){3}}
\put(6.5,1){\line(1,0){1}}
\put(4.5,0.5){\makebox(0,0){$G_{1}$}}
\multiput(11,0)(1,0){3}{\circle*{0.29}}
\multiput(11,1)(1,0){3}{\circle*{0.29}}
\put(11,0){\line(1,0){2}}
\put(11,0){\line(0,1){1}}
\put(12,0){\line(0,1){1}}
\put(13,0){\line(0,1){1}}
\put(12,1){\line(1,0){1}}
\put(10.7,1){\makebox(0,0){$a$}}
\put(10.7,0){\makebox(0,0){$b$}}
\put(11.7,1){\makebox(0,0){$x$}}
\put(11.7,0.3){\makebox(0,0){$u$}}
\put(13.3,1){\makebox(0,0){$y$}}
\put(13.3,0){\makebox(0,0){$v$}}
\put(10,0.5){\makebox(0,0){$G_{2}$}}
\end{picture}\caption{Only $C_{4}$ and $G_{1}$ are very well-covered graphs.}%
\label{fig811}%
\end{figure}

It is easy to prove that every graph having a perfect matching consisting of
pendant edges is very well-covered. The converse is not generally true; e.g.,
the graphs $C_{4}$ and $G_{1}$ depicted in Figure \ref{fig811}. Moreover,
there are well-covered graphs without perfect matchings; e.g., $K_{3}$.
Nevertheless, having a perfect matching\ is a necessary condition for very well-coveredness.

\begin{theorem}
\label{th11}\cite{Favaron}\ For a graph $G$ without isolated vertices the
following are equivalent:

\emph{(i)} $G$ is very well-covered;

\emph{(ii)} there exists a perfect matching in $G$ that satisfies property
\emph{P};

\emph{(iii)} there exists at least one perfect matching in $G$ and every
perfect matching in $G$ satisfies property \emph{P}.
\end{theorem}

\textit{A} matching $M$ in a graph $G$ satisfies \textit{Property} \emph{P} if%
\[
\text{\textquotedblleft}N(x)\cap N(y)=\emptyset\text{, and each }v\in
N(x)-\{y\}\text{ is adjacent to all vertices of }N(y)-\left\{  x\right\}
\text{\textquotedblright}%
\]
hold for every edge $xy\in M$.

For example, the perfect matching $M=\{ab,xy,uv\}$ of the graph $G_{2}$ from
Figure \ref{fig811} does not satisfies Property \emph{P}, since $uv\in M,b\in
N(u),y\in N(v)$, but $by\notin E(G_{2})$. Hence, $G_{2}$ is not a very
well-covered graph. Moreover, $G_{2}$ is not well-covered, because no maximum
stable set of $G_{2}$ includes the stable set $\{b,v\}$. However, $G_{2}$ is a
K\"{o}nig-Egerv\'{a}ry graph. Notice that $K_{4}$ is well-covered, has perfect
matchings, but is neither a K\"{o}nig-Egerv\'{a}ry graph, nor a very
well-covered graph.

\begin{theorem}
\label{th5}\cite{LevMan98} A graph is very well-covered if and only if it is a
well-covered K\"{o}nig-Egerv\'{a}ry graph.
\end{theorem}

A set $A\subseteq V(G)$ is a \textit{local maximum stable set} of $G$ if
$A\in\Omega(G[N[A]])$ \cite{LevMan2}; by $\Psi(G)$ we denote the family of all
local maximum stable sets of the graph $G$. For instance, $\{a\},\{a,e\}\in
\Psi(G)$, while $\{c\},\{b,f\}\notin\Psi(G)$, where $G$ is from Figure
\ref{fig111}. Notice also that in the same graph, the stable sets
$\{a,e\},\{b,f\}$ are contained in some maximum stable sets of $G$, while for
$\{a,c\},\{c,e\}$ this is not true.

\begin{theorem}
\cite{NemhTro}\label{th1} Every local maximum stable set of a graph is a
subset of a maximum stable set.
\end{theorem}

\begin{definition}
\cite{BjZiegler}, \cite{KorLovSch} A \textit{greedoid} is a pair
$(V,\mathcal{F})$, where $\mathcal{F}\subseteq2^{V}$ is a non-empty set system
satisfying the following conditions:

\emph{Accessibility:} for every non-empty $X\in\mathcal{F}$ there is an $x\in
X$ such that $X-\{x\}\in\mathcal{F}$;

\emph{Exchange:} for $X,Y\in\mathcal{F},\left\vert X\right\vert =\left\vert
Y\right\vert +1$, there is an $x\in X-Y$ such that $Y\cup\{x\}\in\mathcal{F}$.
\end{definition}

In the sequel we use $\mathcal{F}$ instead of $(V,\mathcal{F})$, as the ground
set $V$ will be, usually, the vertex set of some graph.

\begin{figure}[h]
\setlength{\unitlength}{1.0cm} \begin{picture}(5,1.5)\thicklines
\multiput(3,0)(1,0){4}{\circle*{0.29}}
\multiput(3,1)(1,0){3}{\circle*{0.29}}
\put(3,0){\line(1,0){3}}
\put(3,1){\line(1,0){1}}
\put(4,1){\line(1,-1){1}}
\put(3,0){\line(0,1){1}}
\put(5,0){\line(0,1){1}}
\put(5.3,1){\makebox(0,0){$f$}}
\put(2.7,0){\makebox(0,0){$a$}}
\put(2.7,1){\makebox(0,0){$b$}}
\put(4.3,0.3){\makebox(0,0){$c$}}
\put(4.3,1){\makebox(0,0){$d$}}
\put(5.3,0.3){\makebox(0,0){$e$}}
\put(6.3,0){\makebox(0,0){$g$}}
\put(2,0.5){\makebox(0,0){$G$}}
\multiput(8.5,0)(1,0){4}{\circle*{0.29}}
\multiput(8.5,1)(1,0){4}{\circle*{0.29}}
\put(8.5,0){\line(1,0){3}}
\put(8.5,0){\line(1,1){1}}
\put(8.5,1){\line(1,0){3}}
\put(10.5,0){\line(0,1){1}}
\put(7.5,0.5){\makebox(0,0){$H$}}
\end{picture}\caption{$\Psi(G)$ is not a greedoid, $\Psi(H)$ is a greedoid.}%
\label{fig2922}%
\end{figure}
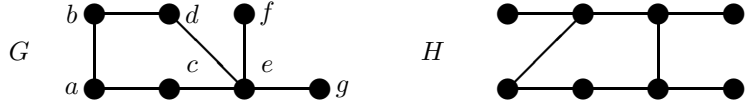

The graphs from Figure \ref{fig2922} are non-bipartite K\"{o}nig-Egerv\'{a}ry
graphs, and all their maximum matchings are uniquely restricted. Let us remark
that both graphs are also triangle-free, but only $\Psi(H)$ is a greedoid. It
is clear that $\{b,c\}\in$ $\Psi(G)$, while $\left\{  b\right\}  ,\left\{
c\right\}  \notin$ $\Psi(G)$. Notice also that $G[N[\{b,c\}]]$ is not a
K\"{o}nig-Egerv\'{a}ry graph, and, as one can see from the following theorem,
this is a good reason for $\Psi(G)$ not to be a greedoid.

\begin{theorem}
\label{th33}\cite{LevMan07} If $G$ is a triangle-free graph, then the
following assertions are equivalent:

\emph{(i)} $\Psi(G)$ is a greedoid;

\emph{(ii)} all maximum matchings of $G$ are uniquely restricted and the
closed neighborhood of every local maximum stable set of $G$ induces a
K\"{o}nig-Egerv\'{a}ry graph.
\end{theorem}

The cases of trees, bipartite graphs, unicycle graphs, whose family of local
maximum stable sets forms a greedoid, were analyzed in \cite{LevMan2},
\cite{LevMan45}, \cite{LevMan09b}, respectively.

In this paper we characterize very well-covered graphs of girth at least four,
whose families of local maximum stable sets are greedoids.

\section{Results}

Let us remark that the very well-covered graph $G_{1}$ in Figure \ref{fig811}
has a $C_{4}$ and one of the edges of this $C_{4}$ belongs to the unique
perfect matching of $G_{1}$.

\begin{lemma}
\label{lem1}No edge of some $C_{q}$, for $q=3$ or $q\geq5$, belongs to a
perfect matching in a very well-covered graph.
\end{lemma}

\begin{proof}
If the graph $G$ is very well-covered, then by Theorem \ref{th11}, $G$ has a
perfect matching, say $M$, and each perfect matching satisfies Property
\emph{P}.

Let $xy\in M$. Then, Property \emph{P} implies that $N(x)\cap N(y)=\emptyset$,
i.e., $xy$ belongs to no $C_{3}$ in $G$. Further, if $v\in N(x)-\{y\}$ and
$u\in N(y)-\{x\}$, Property \emph{P} assures that $vu\in E(G)$, i.e., $xy$
belongs to no $C_{q}$, for $q\geq5$.
\end{proof}

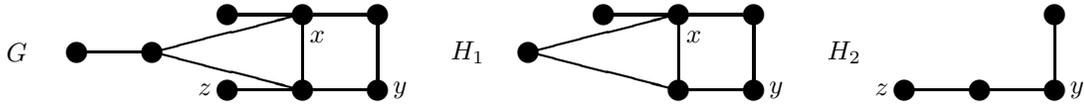
\begin{figure}[h]
\setlength{\unitlength}{1cm}\begin{picture}(5,1.2)\thicklines
\multiput(3,0)(1,0){3}{\circle*{0.29}}
\multiput(3,1)(1,0){3}{\circle*{0.29}}
\multiput(1,0.5)(1,0){2}{\circle*{0.29}}
\put(1,0.5){\line(1,0){1}}
\put(2,0.5){\line(4,1){2}}
\put(2,0.5){\line(4,-1){2}}
\put(3,0){\line(1,0){2}}
\put(3,1){\line(1,0){2}}
\put(4,0){\line(0,1){1}}
\put(5,0){\line(0,1){1}}
\put(4.2,0.7){\makebox(0,0){$x$}}
\put(5.3,0){\makebox(0,0){$y$}}
\put(2.7,0){\makebox(0,0){$z$}}
\put(0.2,0.5){\makebox(0,0){$G$}}
\multiput(9,0)(1,0){2}{\circle*{0.29}}
\multiput(8,1)(1,0){3}{\circle*{0.29}}
\put(7,0.5){\circle*{0.29}}
\put(7,0.5){\line(4,1){2}}
\put(7,0.5){\line(4,-1){2}}
\put(8,1){\line(1,0){2}}
\put(9,0){\line(1,0){1}}
\put(9,0){\line(0,1){1}}
\put(10,0){\line(0,1){1}}
\put(9.2,0.7){\makebox(0,0){$x$}}
\put(10.3,0){\makebox(0,0){$y$}}
\put(6.2,0.5){\makebox(0,0){$H_{1}$}}
\multiput(12,0)(1,0){3}{\circle*{0.29}}
\put(14,1){\circle*{0.29}}
\put(12,0){\line(1,0){2}}
\put(14,0){\line(0,1){1}}
\put(14.3,0){\makebox(0,0){$y$}}
\put(11.7,0){\makebox(0,0){$z$}}
\put(11.2,0.5){\makebox(0,0){$H_{2}$}}
\end{picture}\caption{Both $H_{1}=G[N[\{x,y\}]]$ and $H_{2}=G[N[\{y,z\}]]$ are
K\"{o}nig-Egerv\'{a}ry\ graphs.}%
\label{fig22111}%
\end{figure}

Let us mention that if $G$ is very well-covered, $S$ is a stable set such that
$G[N[S]]$ is a K\"{o}nig-Egerv\'{a}ry graph, then $S$ does not necessarily
belong to $\Psi(G)$; e.g., the set $S_{1}=\{x,y\}$ is stable in the graph $G$
depicted in Figure \ref{fig22111}, and $S_{1}\notin$ $\Psi(G)$, while
$H_{1}=G[N[S_{1}]]$ is a K\"{o}nig-Egerv\'{a}ry graph. Notice that
$S_{2}=\{y,z\}\in\Psi(G)$ and $H_{2}=G[N[S_{2}]]$ is a K\"{o}nig-Egerv\'{a}ry
graph. The following finding, firstly presented in \cite{LevMan07}, shows that
this phenomenon is true for very well-covered graphs in general. We repeat the
proof for the sake of self-containment.

\begin{theorem}
\label{th7}If $G$ is a very well-covered graph, then $G[N[S]]$ is a
K\"{o}nig-Egerv\'{a}ry graph, for every $S\in\Psi(G)$.
\end{theorem}

\begin{proof}
By Theorem \ref{th5}, $G$ is a K\"{o}nig-Egerv\'{a}ry graph. According to
Theorem \ref{th11}, $G$ has a perfect matching, say $M$, and each perfect
matching satisfies Property \emph{P}.

Suppose by way of contradiction that there is $S=\{v_{i}:1\leq i\leq
k\}\in\Psi(G)$, such that $G[N[S]]$ is not a K\"{o}nig-Egerv\'{a}ry graph.

Since $G$ is well-covered, there exists some $W\in\Omega(G)$, with $S\subseteq
W$. By Theorem \ref{th4}, $M\subseteq(W,V(G)-W)$, and because $M$ is a perfect
matching and $S\subseteq W$, we infer that $S$ is matched by $M$ into $N(S)$,
and this implies $\left\vert S\right\vert \leq\left\vert N(S)\right\vert $.
The assumption that $G[N[S]]$ is not a K\"{o}nig-Egerv\'{a}ry graph leads to
$\left\vert N(S)\right\vert >\left\vert S\right\vert $. It means that there
exists a vertex $x\in N(S)-M(S)$, where $M(S)=\{w_{i}:v_{i}w_{i}\in M,1\leq
i\leq k\}$.

In the following, we will prove that the set $\{x\}\cup M(S)$ is stable.

Firstly, $x$ must be adjacent to some vertex, say $v_{1}$, from $S$, otherwise
$S\cup\left\{  x\right\}  $ is a stable set larger than $S$ in $G[N[S]]$, in
contradiction with $S\in\Psi(G)$. By Lemma \ref{lem1}, $x$ is not adjacent to
$w_{1}$, since $v_{1}w_{1}\in M$. Thus, $\left\{  x,w_{1}\right\}  $ is a
stable set.

One of $x,w_{1}$ must be adjacent to one vertex, say $v_{2}$, from $S$,
because, otherwise, the set $\left\{  x,w_{1}\right\}  \cup\{v_{i}:2\leq i\leq
k\}$ would be stable in $G[N[S]]$, larger than $S$. If $w_{1}v_{2}\in E(G)$,
then Property \emph{P}, applied to the edge $v_{1}w_{1}\in M$, ensures that
$xv_{2}\in E(G)$.

In other words, $x$ must be adjacent to $v_{1}$. Moreover, the set
$\{x,w_{1},w_{2}\}$ is stable, because $xw_{2}\notin E(G)$ according to Lemma
\ref{lem1}, while for $w_{1}w_{2}\in E(G)$ we get, by Property \emph{P}, that
$xw_{1}\in E(G)$, in contradiction with the fact that $\{x,w_{1}\}$ is a
stable set.

Assume that for some $j<k$, the set
\[
A_{j}=\{x\}\cup\{w_{i}:1\leq i\leq j\}
\]
is stable, and $x$ is adjacent to each $v_{i},1\leq i\leq j$. Then, there is
an edge joining a vertex, say $a$, belonging to $A$, and a vertex, say
$v_{j+1}$, from the set $\{v_{i}:j+1\leq i\leq k\}$. Otherwise,
\[
A_{j}\cup\{v_{i}:j+1\leq i\leq k\}
\]
is a stable set in $G[N[S]]$, larger than $S$. If $a=w_{t}$, then by Property
\emph{P}, when the edge\emph{\ }$v_{t}w_{t}$ is concerned, the vertex $x$ must
be adjacent to $v_{j+1}$. Thus, no matter where $a$ is located, the vertex $x$
is adjacent to the vertex $v_{j+1}$ (see Figure \ref{fig 33}\textit{(a)}).
\begin{figure}[h]
\setlength{\unitlength}{1cm}\begin{picture}(5,2.2)\thicklines
\multiput(2,0.5)(1,0){3}{\circle*{0.29}}
\multiput(3,1.5)(1,0){2}{\circle*{0.29}}
\multiput(3,0.5)(1,0){2}{\line(0,1){1}}
\multiput(6,0.5)(1,0){3}{\circle*{0.29}}
\multiput(6,1.5)(1,0){3}{\circle*{0.29}}
\multiput(6,0.5)(1,0){3}{\line(0,1){1}}
\multiput(4,0.5)(0.2,0){10}{\circle*{0.11}}
\multiput(6,0.5)(0.2,0){5}{\circle*{0.11}}
\put(2,0.5){\line(1,1){1}}
\put(2,0.5){\line(2,1){2}}
\put(2,0.5){\line(4,1){4}}
\put(2,0.5){\line(5,1){5}}
\put(6,0.5){\line(2,1){2}}
\put(2,0){\makebox(0,0){$x$}}
\put(3,0){\makebox(0,0){$w_1$}}
\put(3,1.9){\makebox(0,0){$v_1$}}
\put(4,0){\makebox(0,0){$w_2$}}
\put(4,1.9){\makebox(0,0){$v_2$}}
\put(6,0){\makebox(0,0){$w_t$}}
\put(6,1.9){\makebox(0,0){$v_t$}}
\put(7,0){\makebox(0,0){$w_j$}}
\put(7,1.9){\makebox(0,0){$v_j$}}
\put(8,0){\makebox(0,0){$w_{j+1}$}}
\put(8,1.9){\makebox(0,0){$v_{j+1}$}}
\put(1.2,1){\makebox(0,0){$(a)$}}
\multiput(10.7,0.5)(1,0){3}{\circle*{0.29}}
\multiput(11.7,1.5)(1,0){2}{\circle*{0.29}}
\multiput(11.7,0.5)(1,0){2}{\line(0,1){1}}
\put(10.7,0.5){\line(1,1){1}}
\put(10.7,0.5){\line(2,1){2}}
\put(11.7,0.5){\line(1,0){1}}
\put(9.8,1){\makebox(0,0){$(b)$}}
\put(10.7,0){\makebox(0,0){$x$}}
\put(11.7,0){\makebox(0,0){$w_t$}}
\put(11.7,1.9){\makebox(0,0){$v_t$}}
\put(12.7,0){\makebox(0,0){$w_{j+1}$}}
\put(12.7,1.9){\makebox(0,0){$v_{j+1}$}}
\end{picture}\caption{\textit{(a)} The vertex $x$ is adjacent to all vertices
from $\{v_{i}:1\leq i\leq j\}$. \textit{(b) }The vertices\textit{\ }
$x,v_{j+1},w_{j+1},w_{t},v_{t}$ span a five vertex cycle.}%
\label{fig 33}%
\end{figure}
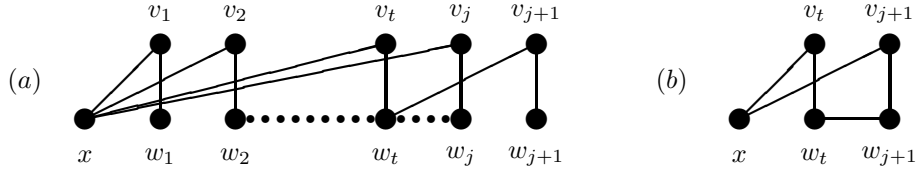

Since $xv_{j+1}\in E(G)$ and $v_{j+1}w_{j+1}\in M$, Lemma \ref{lem1} implies
that the vertices $x$ and $w_{j+1}$ are not adjacent. Moreover, no vertex from
the set $\{w_{i}:1\leq i\leq j\}$ is adjacent to $w_{j+1}$. Otherwise, if some
$w_{t}$ is adjacent to $w_{j+1}$, then $\{x,v_{j+1},w_{j+1},w_{t},v_{t}\}$
spans a five vertex cycle in $G[N[S]]$ (see Figure \ref{fig 33}\textit{(b)}).
In accordance with Property \emph{P}, when the edge\emph{\ }$v_{t}w_{t}$ is
concerned, the vertex $x$ must be adjacent to $w_{j+1}$. Hence, $\{x,v_{j+1}%
,w_{j+1}\}$ spans a triangle, which is impossible, by Lemma \ref{lem1}.

Therefore, the set $A_{j+1}$ is stable. In this way one can eventually reach
the set $\{x\}\cup M(S)$, which must be stable in $G[N[S]]$ like all its
predecessors. Now the inequality
\[
\left\vert \{x\}\cup M(S)\right\vert >\left\vert S\right\vert
\]
stays in contradiction with the following facts:
\[
\{x\}\cup M(S)\subseteq N[S]\text{\ \ and }S\in\Psi(G).
\]

Consequently, $G[N[S]]$ is a K\"{o}nig-Egerv\'{a}ry graph.
\end{proof}

\begin{theorem}
\label{th10}Let $G$ be a very well-covered graph of girth at least $4$. Then
the following assertions are equivalent:

\emph{(i)} $\Psi(G)$ is a greedoid;

\emph{(ii)} $G$ has a unique maximum matching.
\end{theorem}

\begin{proof}
Firstly, Theorem \ref{th11} implies that each maximum matching of $G$ is perfect.

\emph{(i) }$\Longrightarrow$ \emph{(ii) }Since the girth of $G$ is greater or
equal to $4$, the graph $G$\ is triangle-free. Hence, according to Theorem
\ref{th33}, a perfect matching of $G$ is unique.

\emph{(ii) }$\Longrightarrow$ \emph{(i)} In fact, $G$ has a unique perfect
matching. Consequently, every maximum matching of $G$\ is uniquely restricted.
Combining the fact that $G$\ is triangle-free with Theorems \ref{th7} and
\ref{th33}, we conclude that $\Psi(G)$ is a greedoid.
\end{proof}

The structure of very well-covered graphs of girth at least $5$ is more specific.

\begin{theorem}
\cite{DeanZio1994}, \cite{LevMan07a} Let $G$ be a graph of girth at least $5$.
Then $G$ is very well-covered if and only if $G=H\circ K_{1}$, for some graph
$H$ of girth $\geq5$.
\end{theorem}

Consequently, a very well-covered graph of girth $\geq5$ has a unique perfect
matching. Therefore, by Theorem \ref{th10}, we get the following.

\begin{corollary}
\cite{LevMan08}\label{cor1} Each very well-covered graph of girth at least $5$
generates a local maximum stable set greedoid.
\end{corollary}

It is known that the recognition of well-covered graphs is a co-\textbf{NP}%
-complete problem \cite{ChvatalSlater1993}, \cite{SanStew1992}. Nevertheless,
very well-covered graphs can be recognized in polynomial time. Actually, it
goes directly from Favaron's characterization. Namely, to recognize a graph as
being very well-covered, we just need to show that it has a perfect matching
which satisfies property \emph{P}. To find a maximum matching one needs
$O(\left\vert V\right\vert ^{\frac{1}{2}}\bullet\left\vert E\right\vert
)$\emph{ }time \cite{MiVa80}. To check property \emph{P} one has to handle
$O\left(  \left\vert V\right\vert ^{3}\right)  $ pairs of vertices in the
worst case. All in all, it gives us an $O\left(  \left\vert V\right\vert
^{3}\right)  $ algorithm.

If our goal is to recognize very well-covered graphs with unique perfect
matchings, then we may do\ better. The reason for this is that one can test
whether the graph has a unique perfect matching, and find it if it exists, in
$O\left(  \left\vert E\right\vert \bullet log^{4}\left\vert V\right\vert
\right)  $ time \cite{GabKapTar}. Finally, Theorem \ref{th10} and Corollary
\ref{cor1} justify that one can decide in $O\left(  \left\vert E\right\vert
\bullet log^{4}\left\vert V\right\vert \right)  $ time whether $\Psi(G)$ is a
greedoid, for a given very well-covered graph $G$ of girth $\geq4$.

\section{Conclusions}

In this paper we have proved that $\Psi(G)$ is a greedoid for those very
well-covered graphs $G$ of girth $\geq4$ that have a unique perfect matching.

\begin{problem}
Characterize very well-covered graphs of girth three producing local maximum
stable set greedoids.
\end{problem}

\end{document}